\documentclass[twoside]{article}
\newcommand{\mytitle}{Electrostatics on Branching Processes}
\newcommand{\keywords}{Branching process, Boltzmann factor, canonical partition function, grand canonical partition function, rational generating function}
\newcommand{\msc}{
15B52, 
60B05,
60B20, 
60G55,  
60G57,
60J80,
60K37,
82B20,
82B23, 
82B31

}

\usepackage{amsmath,amssymb,
  amsthm,amsfonts,amscd,mathrsfs,pifont,upgreek} 
\usepackage{eucal}  
\usepackage{verbatim}
\usepackage{ifpdf}
\usepackage[colorlinks=true, citecolor=blue]{hyperref}

\ifpdf
        \usepackage[pdftex]{graphicx}
\else
        \usepackage[dvips]{graphicx}
\fi

\newtheorem{thm}{Theorem}[section]
\newtheorem{cor}[thm]{Corollary}

\newtheorem{lemma}[thm]{Lemma}

\newtheorem{thm*}{Theorem}[]
\newtheorem{cor*}[thm*]{Corollary}
\newtheorem{claim*}[thm*]{Claim}
\newtheorem{lemma*}[thm*]{Lemma}
\newtheorem{prop*}[thm*]{Proposition}
\newtheorem{conj*}[thm*]{Conjecture}

\theoremstyle{definition}

\newtheorem{ex}{Example}

\newtheorem{problem*}{Problem}[section]

\newtheorem{question*}{Question}[section]

\newtheorem{defn*}{Definition}

\theoremstyle{remark}

\newcommand{\qq}[1]{\qquad \mbox{#1} \qquad}

\newcommand{\BB}[1]{\ensuremath{\mathbb{#1}}}
\newcommand{\E}{\ensuremath{\BB{E}}}

\newcommand{\N}{\ensuremath{\BB{N}}}
\newcommand{\Q}{\ensuremath{\BB{Q}}}
\newcommand{\R}{\ensuremath{\BB{R}}}
\newcommand{\Z}{\ensuremath{\BB{Z}}}
\newcommand{\C}{\ensuremath{\BB{C}}}

\newcommand{\piecewise}[1]{
\left\{
\begin{array}{ll}
#1
\end{array}
\right.
}

\usepackage[margin=1 in]{geometry}
\usepackage{txfonts}

\numberwithin{equation}{section}

\numberwithin{equation}{section}

\begin{document}
\title{\bfseries\sffamily \mytitle}
\author{\sc Christopher D.~Sinclair}
\maketitle

\begin{abstract}
We introduce a random probability measure on the profinite completion of the random tree of a branching process and introduce the canonical and grand canonical ensembles of random repelling particles on this random profinite completion at inverse temperature $\beta > 0$. We think of this as a random spatial process of particles in a random tree, and we introduce the notion of the {\em mean} canonical and grand canonical partition functions where in this context `mean' means averaged over the random environment. We give a recursion for these mean partition functions and demonstrate that in certain instances, determined by the law for the branching process, these partition functions as a function of $\beta$ have algebraic properties which generalize those that appear in the non-random and $p$-adic environments.
\end{abstract}

{\bf MSC2010:} \msc

\vspace{1cm}

\section{Introduction}

The partition function of a physical system is (more-or-less) the Laplace transform of the distribution of allowable energies in the system. It is a function of the basic parameters of the system, and in many instances thermodynamic potentials for a system can be derived from a closed form for the partition function.

Here we consider systems of particles on the profinite completion of a branching process. Precise definitions will follow, but for now, we may think of the environment in which the particles live as a rooted infinite tree, where we think of a particle as an infinite path (a sequence of descendant vertices) from the root down the tree. In two dimensional ($\R^2$) electrostatics, the energy contribution from two (repelling) like charged particles is proportional to the log of the distance between them, and we investigate an analogous situation on random trees by introducing a natural distance function and then assuming that the total energy of a system of particles is the sum of the contributions from all pairs of particles.

When put in contact with a heat reservoir (and/or particle reservoir) our system can exchange energy (and/or particles) and the physical parameters of our system will include the `inverse temperature' (classically denoted by $\beta$), which is the conjugate variable to energy in the Laplace transform. The situation where the number of particles $N$ is fixed is called the {\em canonical ensemble} and the partition function is a function in the variable $\beta$. Canonical partition functions for idealized systems (like the toy model we are introducing here) often have remarkable analytic properties in $\beta$ and an understanding of, for instance, the singularities can lead to an understanding of quantities of physical interest, for instance phase transitions.

When our system is in contact with a heat and particle reservoir, then there is an additional relevant physical parameter, the {\em fugacity} $t$ which is the conjugate variable to $N$. The variable particle situation is called the {\em grand canonical ensemble}, and the analytic properties of the grand canonical partition function in $\beta$ and $t$ can provide useful information about physically relevant quantities (for instance the expected number of particles in the system). The grand canonical partition function can also be seen as the generating function for the canonical partition functions and so understanding the latter can with the former.

Recent work in $p$-adic electrostatics \cite{MR4218379, MR4419249, MR4048290, MR4365943, MR4586932} has introduced combinatorial methods for recursively evaluating the canonical partition function when particles are constrained to the (infinite) complete $p$-nary tree. The principal observation is a functional equation for the grand canonical partition function, which leads to recursions for the canonical partition function. Among the interesting implications of this is that the canonical partition function is a rational function in $p^{-\beta}$ and the poles of this rational function are highly structured. Another implication discovered from this rational function is that the system is still well-defined for a narrow range of {\em negative} temperatures. The negative temperature regime can be reinterpreted as a collection of attracting particles at a high enough temperature where the thermal energy overwhelms the attraction so that they don't collocate (an infinite energy situation). The abscissa of analyticity of the canonical partition function tells us the critical temperature at which the attraction overcomes the thermal energy.

We will produce similar results in the situation considered here. There is one main difference however; our environment is random. In order to get deterministic partition functions we have to take expectations with respect to the law determining the tree to produce the {\em mean} partition functions. When the tree is formed from a simple branching process (each node having an independent and identically distributed number of children), this expectation becomes tractable and we arrive at functional equations and recursions for the mean grand canonical and canonical partition functions. The analytic properties of these are interesting, and in particular if the possible number of children for any node is finite, the mean canonical partition function is a rational function in $\{ q^{-\beta} : \mathbb P(q) > 0\}$. For instance, if the only possible number of children of any node is 5 or 2, then the canonical partition function is a rational function in $2^{-\beta}$ and $5^{-\beta}$. Branching processes are a well-studied class of processes \cite{MR1991122}.

As a final connection, the random and mean canonical partition functions for our branching process are related to special examples of Igusa zeta function. In their generality Igusa zeta functions are natural integrals involving absolute values of polynomials over non-archimedean spaces \cite{MR1743467}. These arise as Euler factors in certain kinds of global zeta functions, and it is known that if the residue class degree of the non-archimedean space is $q$, then the Igusa zeta function is a function in $q^{-\beta}$ \cite{MR347753, MR404215}. The results here, in which we get rational functions in multiple $q^{-\beta}$ present a possible interesting extension of Igusa zeta functions to a larger class of integrals. Here the polynomial that arises in the $p$-adic electrostatic situation is the absolute Vandermonde determinant. There is not a natural definition of polynomials on trees, and in our case we will generalize the situation by replacing the absolute Vandermonde with a product of pairwise distances between particles on our trees. Integrals with Vandermonde determinants appear all over mathematics \cite{MR2434345} as the Selberg Integral \cite{MR0018287}, in multivariate statistics \cite{MR652932}, volumes of matrix groups, volumes of polynomials \cite{MR1868596, sinclair-2005, pet-sin, MR2145532}, in random matrix theory \cite{mehta:256, MR0173726, sinclair-2007}, as $p$-adic string amplitudes \cite{MR4373375, MR3946489}, etc.

\section{Profinite Completions of Trees}

Let $T$ be an infinite rooted tree with vertex set $V$. We will denote the root $v_0$, and if $w \in V$ is a descendant of $v \in V$ we will write $w | v$ ($w$ descends from $v$). If $w$ is an immediate descendant (child) of $v$, then we will write $w \| v$. We also define $Q(v)$ to be the number of children of $v$, and we will assume that $Q(v)$ is finite for all $v$. This is an example of an infinite Galton-Watson process \cite{Galton-Watson}.

A chain of descendants is a sequence of vertices $\mathbf v = (v_n); n=0,1,2,\ldots$ starting at the root with $v_{n+1} \| v_n$. If the vertex $v$ appears in $\mathbf v$ we will write $\mathbf v | v$ (so, for instance $\mathbf v | v_n$ for all $n$).

The set of descendant chains is the {\em profinite completion} of $T$, and we will denote it $\mathcal T$. We may equip $\mathcal T$ with a $\sigma$-algebra (and topology) by defining for each $v \in V$ the open set
\[
\mathcal T(v) = \{ \mathbf v \in \mathcal T : \mathbf v | v \}.
\]
We then take $\mathcal H$ to be the $\sigma$-algebra generated by $\{\mathcal T(v) : v \in V\}$. The topology generated by $\{\mathcal T(v) : v \in V\}$ is totally disconnected.

We write $\mathbf a(v) = (a_0(v), a_1(v), \ldots, a_{d(v)-1}(v))$ for the ancestor chain of $v$,
\[
  v \|  a_{d(v) - 1}(v) \| \cdots \| a_1(v) \| a_0(v) .
\]
Because $T$ is rooted, we necessarily have $a_0(v) = v_0$. Here $d(v)$ is the {\em generation} or {\em depth} of $v$.

We may then define the measure $\mu$ on $(\mathcal T, \mathcal H)$ by specifying
\[
\mu(\mathcal T(v)) = \prod_{a \in \mathbf a(v)} \frac{1}{Q(a)}.
\]
It is easy to see that $\mu(\mathcal T) = 1$ and hence $\mu$ is a probability measure. We may think of this measure as that determined by equal inheritance among siblings. In this interpretation,  $\mu(\mathcal T(v))$ the amount of inheritance of $v$ to be distributed among the $Q(v)$ children of $v$ equally. It is also easy to see that $\mu$ is a measure; because $\mathcal H$ is totally disconnected, each $\mathcal T(v)$ is the disjoint union of the $\mathcal T(w)$ where $w \| v$, and thus an arbitrary probability measure $\nu$ is essentially an inheritance scheme where $\nu(\mathcal T) = 1$ and for all vertices $v$, $\nu(\mathcal T(v)) = \sum_{w \| v}\nu(\mathcal T(w))$.

\begin{lemma}
\label{ultrametric}
If $\mu$ has no atoms, $\delta: \mathcal T \times \mathcal T \rightarrow [0,1]$, given by
\[
\delta(\mathbf v,\mathbf w) = \inf\{ \mu(\mathcal T(v)) : v \in V, \mathcal T(v) \ni \mathbf v, \mathbf w \}
\]
is an ultrametric. The infimum is over all $v \in V$ for which $\mathcal T(v)$ contains both $\mathbf v$ and $\mathbf w$. If $\mu$ has atoms then $\delta$ is a pseudo-ultrametric.
\end{lemma}
\begin{proof}
We only need to prove that $\delta$ satisfies the ultrametric inequality. That is, if $\mathbf v, \mathbf w, \mathbf z \in \mathcal T,$ then
\[
\delta(\mathbf v, \mathbf w) \leq \max\{\delta(\mathbf v, \mathbf z), \delta(\mathbf z, \mathbf w)\}.
\]
Let $v, v', v'' \in V$ be the relevant least common ancestors such that
\[
\delta(\mathbf v, \mathbf w) = \mu(\mathcal T(v)); \qquad
\delta(\mathbf v, \mathbf z) = \mu(\mathcal T(v')); \qquad
\delta(\mathbf z, \mathbf w) = \mu(\mathcal T(v'')).
\]
There are two cases, the first when $\mathbf z \in \mathcal T(v)$. In this case, $v' | v$ and $v'' | v$ and because $v$ is the least common ancestor of $\mathbf v$ and $\mathbf w$, we must have either $v = v'$ or $v = v''$.

The other case is when $\mathbf z \not \in \mathcal T(v)$. Note that because $v'$ is the least common ancestor $\mathbf v$ and $\mathbf z$, $v | v'$, in which case
\[
\delta(\mathbf v, \mathbf w) = \mu(\mathcal T(v)) \leq \mu(\mathcal T(v')) = \delta(\mathbf v, \mathbf z),
\]
and the lemma follows by symmetry.
\end{proof}
Notice in particular, if $v$ and $w$ are different children of $v_0$ (assuming it has more than one child) and $\mathbf v | v$ and $\mathbf w | w$, then $\delta(\mathbf v, \mathbf w) = 1$. In the inheritance interpretation, $\delta(\mathbf v, \mathbf w)$ is the wealth of the least wealthy common ancestor of $\mathbf v$ and $\mathbf w$.

\includegraphics[width=5in]{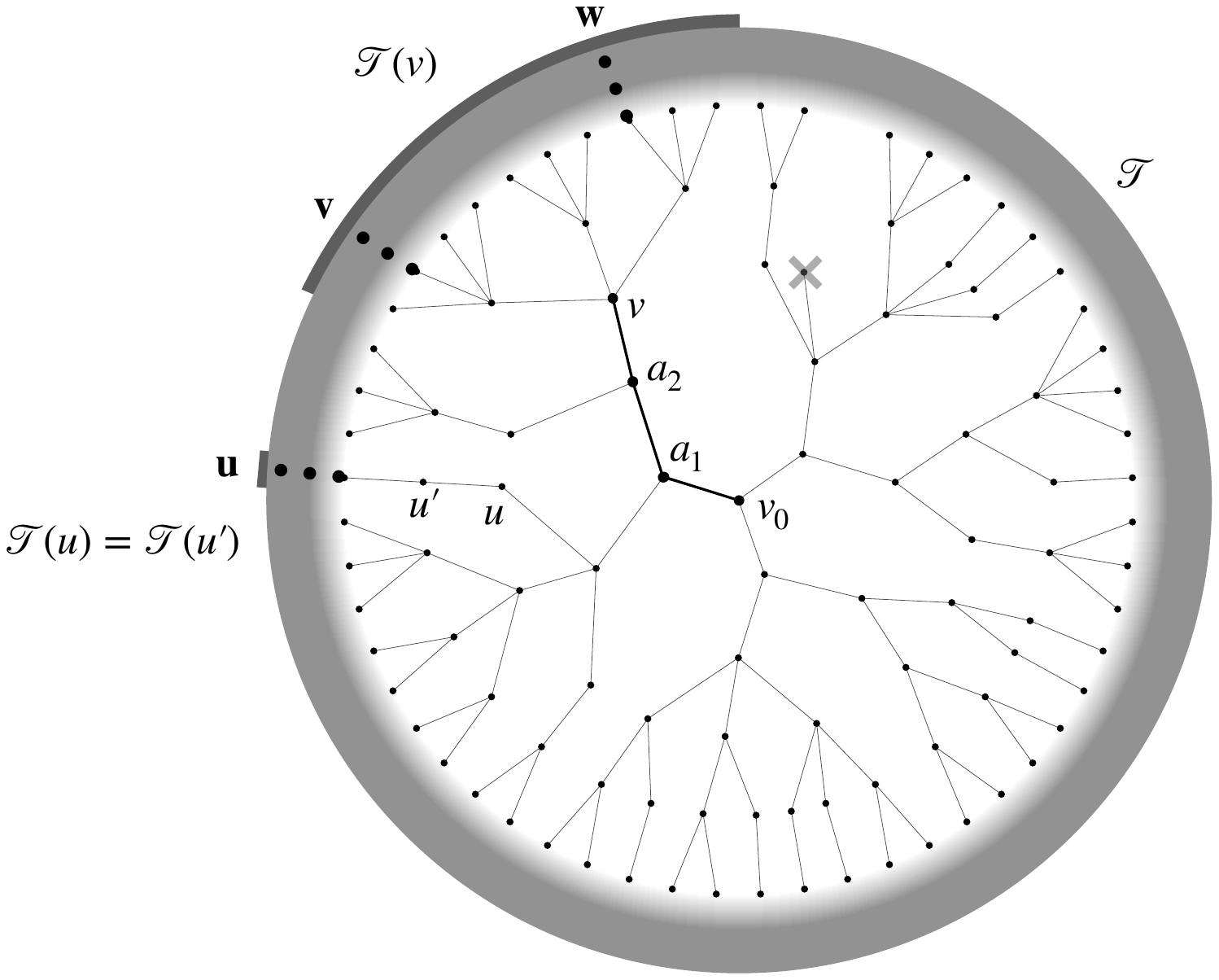}

In the example pictured, we view $\mathcal T$ as the `boundary' of the infinite tree (of which only the first few generations are shown). We will assume that all vertices have at least one (and hence infinitely many) descendant(s) so that vertices like the one crossed out do not appear. If we have two points $\mathbf v, \mathbf w \in \mathcal T$, then the distance between them is the measure of the set of descendants of their least common ancestor, represented here as the part of the boundary labelled $\mathcal T(v)$. Because $Q(v_0) = 3$, $Q(a_1) = Q(a_2) = 2$ that $\delta(\mathbf v, \mathbf w) = \mu(\mathcal T(v)) = 1/12$. Notice also that $\mu(\mathcal T(u)) = \mu(\mathcal T(u')) = 1/18$. It is possible, if all descendants of $u'$ have exactly one child, that
$\delta(\mathbf u, \mathbf u) =\mu\{ \mathbf u \} = \mu(\mathcal T(u))> 0$. In this situation, $\mu$ has an atom, and $\delta$ is not a metric (though it still satisfies the strong triangle inequality). In the situations we will be most interested we will assume that each vertex has positive probability of having more than one child. Thus, with probability one, the measures $\mu$ we will be dealing with have no atoms in $\mathcal T$ and the $\delta$ will be true ultrametrics. When $\mu$ has no atoms, we may replace the tree with that formed by contracting all edges between parents and only-children. There is an obvious bijection between $\mathcal T$ and the profinite completion of this new contracted tree, and the measure and metrics are invariant under this bijection. Thus, we may assume that all trees of interest have the property that each vertex has at least two children. Alternately, we may insert ancestor chains of the form $\multimapdotboth \cdots \multimapdotboth$, of arbitrary but finite length, between any pair of adjacent vertices without functionally changing $\mu$ and $\delta$. By taking scaling limits this allows us to consider measures and metrics on the profinite completion of trees formed from branching processes where the branching happens at arbitrary real (deterministic or random) times.

In the special case where $Q(v) = q$ for all $v \in V$, $\mathcal T$ is the profinite completion of the (infinite) rooted regular $q$-nary tree. When $q=p$ is a rational prime, $\mathcal T$ can be identified with the $p$-adic integers $\Z_p$ and in this case the measure $\mu$ is Haar measure on $\Z_p$ and the distance $\delta$ is that inherited from the $p$-adic absolute value. A quick reference for salient features of the $p$-adic numbers is \cite[\S3.3]{MR4175370}.When $q$ is a prime power, similar identifications can be made with the ring of integers of nonarchimedean completions of number fields (see \cite[Part I]{MR554237}). The nonarchimedean context  \cite{MR4419249, MR4326614} is the deterministic precursor to the results presented here and after a few definitions we will return to it to see what theorems may be generalized to the random situation.

\section{Electrostatics on Profinite Completions}

In classical two-dimensional electrostatics, the potential energy of two like charged particles is assumed to be proportional to the (negative) logarithm of the distance between them. A complete reference for the classical two-dimensional situation is \cite{forrester-book}. Generalizing this setup to the present situation, we define the potential energy of a system of two particles in $\mathbf v, \mathbf w \in \mathcal T$ by
\[
E(\mathbf v, \mathbf w) = -\log \delta(\mathbf v, \mathbf w).
\]
If we have $N$ like charged particles located at the coordinates of $\vec{\mathbf v} = (\mathbf v_1, \ldots, \mathbf v_N)$ then the total potential energy caused by the interaction between particles is the sum of the potential energies of all pairs of particles. That is,
\[
E(\vec{\mathbf v}) = -\sum_{m < n} \log \delta(\mathbf v_n, \mathbf v_m).
\]
The fact that $E(\vec{\mathbf v}) = +\infty$ if any two coordinates of $\vec{\mathbf v}$ are equal, means that the particles repel. This is sometimes called a {\em plasma} in analogy with a gas of electrons. It is sometimes necessary to introduce a background potential to keep the system from minimizing the energy by allowing all distances $\delta(\mathbf v_n, \mathbf v_m)$ to go to infinity (an uninteresting situation) however, in our case $\delta \leq 1$ and so no background potential is necessary to prevent this.

Under the assumptions of Boltzmann statistics \cite{MR116523}, if we put our system in contact with a heat bath (energy reservoir) at inverse temperature $\beta$, the probability density of the system being in state $\vec{\mathbf v}$ is given by
\[
\frac{1}{N! Z(N,\beta)} e^{-\beta E(\vec{\mathbf v})} d\mu^N(\vec{\mathbf v}) = \frac{1}{N! Z(N,\beta)}\bigg\{ \prod_{m<n}^N \delta(\mathbf v_n ,\mathbf v_m)^{\beta} \bigg\} \, d\mu^N(\vec{\mathbf v}),
\]
where $\mu^N$ is the $N$-fold product measure of $\mu$ on $\mathcal T^N$ and $Z(N,\beta)$ is the {\em canonical partition function} on $\mathcal T$, given explicitly by
\[
Z(N,\beta) = \frac{1}{N!} \int_{\mathcal T^N} \bigg\{ \prod_{m<n}^N \delta(\mathbf v_n ,\mathbf v_m)^{\beta} \bigg\} \, d\mu^N(\vec{\mathbf v}).
\]
The mathematical determination of the partition function (in terms of special functions, a recursion, etc) is often the first step to `solving' the system.

The {\em grand} canonical ensemble arises when we allow the number of particles to vary; that is, when the system is put in contact with an energy {\em and} particle reservoir. In this situation there is a parameter $t$, the {\em fugacity}, which arises as a sort of energy-cost-per-particle, and Boltzmann statistics in this situation tell us that the conditional probability density of our system being in state $\vec{\mathbf v}$, conditioned on there being $N$ particles is
\[
\frac{1}{Z(t, \beta)} \frac{t^N}{N!} \bigg\{ \prod_{m<n}^N \delta(\mathbf v_n ,\mathbf v_m)^{\beta} \bigg\} \, d\mu^N(\vec{\mathbf v}),
\]
where the {\em grand} canonical partition function $Z(t,\beta)$ is given by
\[
Z(t,\beta) = \sum_{N=0}^{\infty} Z(N,\beta) t^N.
\]
Notice that $Z(t, \beta)$ can also be arrived at as the generating function for the $Z(N, \beta)$.
We remark that all quantities that appear in this section are dependent on $\mathcal T$. This is important, because when we make $\mathcal T$ random, all of these quantities become random too.

\subsection{The $p$-adic plasma}

To provide context for our results in the random situation, we look at what sorts of results are known in the $p$-adic case and how these might be generalized. In the language we have developed here the results in this section are valid for all regular $q$-nary trees---for any integer $q > 1$---though strictly speaking the identification with non-archimidean completions of a number field (the generalized $p$-adic case) is only valid when $q$ is a prime power.

Let us denote the canonical and grand canonical partition functions in this situation by $Z_q(N, \beta)$ and $Z_q(t, \beta)$.

First discovered was a recursion in $N$ for $Z_q(N, \beta)$ \cite{sinclair2020nonarchimedean}. This recursion, is a bit messy, but can be reported as
\[
Z_q(N, \beta) = \sum_{N_1 + \cdots + N_q = N} \prod_{\ell=1}^q q^{-\beta{N_{\ell} \choose 2} - N_{\ell}} Z_q(N_{\ell}, \beta).
\]
The sum is over all partitions of $N$ into non-negative integers $N_1, \ldots, N_q$. It should be remarked that $Z_N$ appears on both sides of this equation, but with different coefficients, so this leads to an honest recursion for $Z_q(N, \beta)$.

It is not too hard to use this recursion to show that the grand canonical partition functions satisfies a functional equation with a related generating functions. Specifically, if we define
\[
F_q(t, \beta) = \sum_{N=0}^{\infty} Z(N, \beta) q^{-\beta{N \choose 2}} t^N \qq{then}
Z_q(t, \beta) = F_q(t/q, \beta)^q.
\]
This is the ``$q$-power identity". By taking logarithmic derivatives, the $q$-power identity produces a quadratic recursion for $Z_q(N, \beta)$ given by
\[
\sum_{n=0}^N \left(\frac{N}{q+1} - n)\right) q^{-\beta {n \choose 2}} Z_q(n, \beta) Z_q(N-n, \beta) = 0.
\]

The rest of the paper will be dedicted to generalizing these three results to random profinite completions of branching processes.

\section{The Canonical Ensemble on a Random Profinite Completion}

Let $\mathcal T$ to be the profinite completion of the tree $T$ given by a random branching process. If we denote the set of vertices of $T$ by $V_T$, then our assumption is that the $\{Q(v) : v \in V_T\}$ are independent and identically distributed. Taking $Q = Q(v_0)$ as our canonical representative of this distribution, the law of $\mathcal T$ is then completely determined by
\[
p_q = \mathbb P\{ Q = q \}; \quad q=0,1,\ldots.
\]
We will assume the $p_0 = 0$, $p_1 \neq 1$. This is equivalent to each parent having at least one child, and positive probability of having more than one child. At times it will be useful to assume that $Q$ has finite mean (or perhaps, even more strongly, be bounded). We will call $\mathcal T$ the {\em profinite completion of the branching process with law determined by} $Q$. 

Because $\mathcal T$ is random, $Z(N,\beta)$ is a random function and we therefore introduce the deterministic {\em mean} canonical partition function, by
\[
\overline Z(N,\beta) =  \mathbb E[Z(N,\beta)] = \mathbb E\bigg[\frac{1}{N!} \int_{\mathcal T^N} \bigg\{ \prod_{m<n}^N \delta(\mathbf v_n ,\mathbf v_m)^{\beta}\bigg\} \, d\mu^N(\vec{\mathbf v})\bigg].
\]
The expectation is taken over all $\mathcal T$ with law specified by $Q$. In some cases we write $\overline Z_N(\beta)$ for $\overline Z(N,\beta)$.

\begin{lemma}
\[
\mathbb E\bigg[Z_N(\beta) \bigg| Q = q\bigg] = \sum_{N_1 + \cdots + N_q = N} \prod_{k=1}^q  \mathbb E\bigg[\frac{1}{N_k!} \int_{\mathcal T(w_k)^{N_k}} \bigg\{ \prod_{m<n}^{N_k} \delta(\mathbf v_n ,\mathbf v_m)^{\beta}\bigg\} \, d\mu^{N_k}(\vec{\mathbf v}) \bigg| Q = q
\bigg].
\]
\end{lemma}
\begin{proof}
The event $\{Q = q \}$ is that on which $v_0$ has $q$ children, which we can label $w_1, \ldots, w_q$. If $\mathbf v \in \{ Q = q \}$ then $\mathbf v \in \mathcal T(w_k)$ for exactly one of $w_1, \ldots, w_q$. We can partition the elements of $\vec{\mathbf v}$ by which of the $\mathcal T(w_k)$ they fall into. That is, we can partition the state space based on occupation numbers $N_1, N_2, \ldots, N_q$ of $\mathcal T(w_1), \mathcal T(w_2), \ldots \mathcal T(w_q)$. By necessity $N_1 + N_2 + \cdots + N_q = N$ and moreover, a general state in $\mathcal T(w_1)^{N_1} \times \cdots \mathcal T(w_q)^{N_q}$ describes ${N \choose N_1, \ldots N_q}$ states in $\mathcal T^N$. Finally, if $w_j \neq w_k$ and $\mathbf v \in \mathcal T(w_j)$ and $\mathbf w \in \mathcal T(w_k)$, then $\delta(\mathbf v, \mathbf w) = 1$. In which case Fubini's Theorem allows us to exchange the product and the expectation (on the domain of analyticity in $\beta$).
\end{proof}

\begin{lemma}
  \label{lemma:1}
\[
\mathbb E\bigg[\frac{1}{N_k!}
\int_{\mathcal T(w_k)^{N_k}} \bigg\{ \prod_{m<n}^{N_k} \delta(\mathbf v_n ,\mathbf v_m)^{\beta}\bigg\} \, d\mu^{N_k}(\vec{\mathbf v}) \bigg| Q = q
\bigg] = q^{-N_k -\beta{N_k \choose 2}} \overline Z_{N_k}(\beta).
\]
\end{lemma}
\begin{proof}
This follows immediately by observing that if we define $\mu_k$ and $\delta_k$ to be the measure and distance function on $\mathcal T(w_k)$ (rooted at $w_k$), then
\[
\delta_k(\mathbf v_n ,\mathbf v_m) = \frac{1}{q} \delta(\mathbf v_n ,\mathbf v_m) \qq{and} \mu_k(dv) = \frac{1}{q} \mu(dv).
\]
and because $\mathcal T(w_k)$ has the same law as $\mathcal T$,
\[
\mathbb E\bigg[\frac{1}{N_k!}
\int_{\mathcal T(w_k)^{N_k}} \bigg\{ \prod_{m<n}^{N_k} \delta(\mathbf v_n ,\mathbf v_m)^{\beta}\bigg\} \, d\mu^{N_k}(\vec{\mathbf v}) \bigg| Q = q
\bigg] = q^{-N_k -\beta{N_k \choose 2}} \E\bigg[\frac{1}{N_k!}
\int_{\mathcal T^{N_k}} \bigg\{ \prod_{m<n}^{N_k} \delta(\mathbf v_n ,\mathbf v_m)^{\beta}\bigg\} \, d\mu^{N_k}(\vec{\mathbf v}) \bigg],
\]
and the lemma follows.
\end{proof}
This provides a recursion for $\overline Z_N$, since we can isolate the terms on the right-hand-side which depend on $\overline Z_N$ and all remaining terms depend on $\overline Z_{N_k}$ where the $N_k$ are necessarily less than $N$. This leads to the following theorem.
\begin{thm}
$\{\overline Z_N(\beta) : N \in \mathbb N\}$ satisfy the recursion
  \[
  \overline Z_N(\beta) = \left(1 - \mathbb E\left[ Q^{1 - N - \beta{N \choose 2}}\right]
  \right)^{-1} \sum_{q} p_q \sum_{N_1 + \cdots + N_q = N \atop N_k \neq N} \prod_{k=1}^q q^{-N_k -\beta{N_k \choose 2}} \overline Z_{N_k}(\beta).
  \]
  where on the right hand side the inner sum is over $q$-tuples of non-negative integers $(N_1, \ldots, N_q)$ which sum to $N$, but for which none of the $N_k = N$; with initial condition $\overline Z_0(\beta) = \overline Z_1(\beta) = 1$.
\end{thm}
Note that, we can expand the expectation so that
\[
\overline Z_N(\beta) =   \left(1 - \mathbb E\left[ Q^{1 - N - \beta{N \choose 2}}\right] \right)^{-1} \E\bigg[ \sum_{N_1 + \cdots + N_Q = N \atop N_k \neq N} \prod_{k \leq Q} Q^{-N_k - \beta{N_k \choose 2}} \overline Z_{N_k}(\beta) \bigg].
\]
\begin{proof}
  We may write $\overline Z_N(\beta)$ as a sum over the events $\{Q = q \}$. Specifically,
  \[
    \overline Z_N(\beta) =  \sum_{q} p_q \mathbb E\bigg[\frac{1}{N!} \int_{\mathcal T^N} \bigg\{ \prod_{m<n}^N \delta(\mathbf v_n ,\mathbf v_m)^{\beta}\bigg\} \, d\mu^N(\vec{\mathbf v}) \bigg| Q = q\bigg],
  \]
  and the theorem follows from Lemma~\ref{lemma:1} upon solving for $\overline Z_N(\beta)$.
\end{proof}

\begin{cor}
Suppose there is an integer $K$ so that $Q \leq K$, then $\overline Z_N(\beta)$ is a rational function in $\Q(p_1, p_2, \ldots, p_K, 2^{-\beta}, \ldots, K^{-\beta})$. More specifically, in this situation, $\overline Z_N(\beta)$ is a rational function in $\Q(p_k, k^{-\beta} : p_k \neq 0)$.
\end{cor}
\begin{proof}
Induct on $N$, and note that $\mathbb E[ Q^{1 - N - \beta{N \choose 2}}]$ is a polynomial in $\Q[p_k, k^{-\beta} : p_k \neq 0]$.
\end{proof}

\begin{ex}The first non-trivial example is when $Q$ takes two possible values. To connect this the partition functions arising from electrostatics on $p$-adic fields studied in \cite{sinclair2020nonarchimedean}, consider the case where
\[
\mathbb P\{Q = q\} = p, \qquad  \mathbb P\{Q = 1\} = (1-p).
\]
In this situation,
\[
\mathbb E\left[ Q^{1 - N - \beta{N \choose 2}}\right] = (1-p) + p q^{1 - N - \beta{N \choose 2}}.
\]
Notice also that, when $q = 1$, the summand is 0, and it follows that,
\[
\overline Z_N(\beta) = \left(1 - q^{1 - N - \beta{N \choose 2}}\right)^{-1} \sum_{N_1 + \cdots + N_q = N \atop N_k \neq N} \prod_{k=1}^q q^{-N_k -\beta{N_k \choose 2}} \overline Z_{N_k}(\beta),
\]
which, when $q$ is a rational prime, is the same recursion (with the same initial conditions) satisfied by
\[
\frac{1}{N!}\int_{\Z_q^N} \bigg\{ \prod_{m<n} |x_n - x_m|_q^{\beta} \bigg\} \, d\mathbf x,
\]
where $d\mathbf x$ represents integration with respect to Haar probability measure on $\Z_q^N$. Thus in this example,
\[
\overline Z_N(\beta) = \frac{1}{N!}\int_{\Z_q^N} \bigg\{ \prod_{m<n} |x_n - x_m|_q^{\beta} \bigg\} \, d\mathbf x,
\]
which is the canonical partition function in the $p$-adic situation.
\end{ex}

\section{Quadratic Recurrences}

In this section we are going to consider the rooted tree formed by connecting two infinite rooted trees $\mathcal T$ and $\mathcal P$ together by adding an edge from the root $v_0$ of one copy to the root of the other $w_0$. Eventually we will view $\mathcal T$ and $\mathcal P$ as being independent copies of a random tree with the same law determined by $Q = Q(v_0)$. For now we assume that $v_0$ has $q$ children in $\mathcal T$. Connecting $\mathcal P$ to $v_0$ essentially adds $w_0$ as a child of $v_0$. Let us denote the profinite completion of this new tree by $\mathcal U$. We will use $\mu_{\mathcal U}, \mu_{\mathcal T}, \mu_{\mathcal P}$ and
to represent the measures $\mathcal U, \mathcal T$ and $\mathcal P$ respectively. Notice also that we may identify $\mathcal U(w_0)$ with $\mathcal P$, and so we may write
\[
\mathcal U = \mathcal T \sqcup \mathcal P.
\]
Notice however that, $\mathcal P$ is subordinate to $\mathcal T$ in that the latter has more measure. In particular,
\[
\mu_{\mathcal U}(\mathcal T) = \frac{q}{q+1}, \qq{and} \mu_{\mathcal U}(\mathcal P) = \frac{1}{q+1}.
\]
Notice also that
\[
\mathcal U = \mathcal P \sqcup
\bigsqcup_{v \| v_0 \atop v \neq w_0} \mathcal T(v) = \bigsqcup_{v \| v_0} \mathcal U(v),
\]
and
\[
\mu_{\mathcal U} \mathcal U(v) = \frac{1}{q+1} \qq{for all} v \| v_0.
\]

We will also denote the ultrametrics on $\mathcal T$ and $\mathcal P$ by $\delta_{\mathcal T}$ and $\delta_{\mathcal P}$, though we will construct a function $\Delta_{\mathcal U}$, different from $\delta_{\mathcal U}$, to play the role of the ultrametric on $\mathcal U$. Specifically, we define
\[
\Delta_{\mathcal U}(\mathbf x, \mathbf y) = \piecewise{1 & \mathbf x \in \mathcal T, \mathbf y \in \mathcal P; \\
\delta_{\mathcal T}(\mathbf x,\mathbf y) & \mathbf x, \mathbf y \in \mathcal T; \\
\frac{1}{q} \delta_{\mathcal P}(\mathbf x,\mathbf y) & \mathbf x, \mathbf y \in \mathcal P.}
\]
$-\log \Delta$ forms a perfectly good energy functional, and in this situation, Boltzmann statistics say the probability density of $N$ particles in $\mathcal U$ being in state $\vec{\mathbf x}$ is given by
\begin{equation}
  \label{Boltzmann factor}
\frac{1}{N! Z_{N, \mathcal U}(\beta)} \bigg\{ \prod_{m<n} \Delta_{\mathcal U}(\mathbf x_n, \mathbf x_m)^{\beta}\bigg\} \, d\mu_{\mathcal U}^N(\vec{\mathbf x}),
\end{equation}
where the canonical partition function for the system is
\[
Z_{N, \mathcal U}(\beta) = \frac{1}{N!} \int_{\mathcal U^N} \bigg\{ \prod_{m<n} \Delta_{\mathcal U}(\mathbf x_n, \mathbf x_m)^{\beta}\bigg\} \, d\mu_{\mathcal U}^N(\vec{\mathbf x}).
\]

Let $N_{\mathcal P}(\vec{\mathbf x})$ be the number of coordinates of $\mathbf x$ in $\mathcal P \subset \mathcal U$. We will view $N_{\mathcal P}$ as a random variable, where the distribution of $\vec{\mathbf x}$ is distributed as in \eqref{Boltzmann factor}.

\begin{thm}
  \label{thm: quadratic recursion}
  \[
  \sum_{n=0}^N (\E[N_{\mathcal P}] - n) q^{-\beta{n \choose 2}-n} Z_{\mathcal P,n}(\beta) \, Z_{\mathcal T, N-n}(\beta) = 0,
  \]
where $Z_{\mathcal P,N}$ and $Z_{\mathcal T, N}$ are the canonical partition functions for $\mathcal P$ and $\mathcal T$ respectively.
\end{thm}
\begin{proof}
\[
\E[N_{\mathcal P}] = \frac{1}{Z_{N, \mathcal U}(\beta)}\sum_{n=0}^{N} \frac{1}{(n-1)!} \int_{\mathcal P^n} \bigg\{ \prod_{j<k}^n \Delta_{\mathcal U}(\mathbf w_k, \mathbf w_j)^{\beta}\bigg\} \, d\mu_{\mathcal U}^n(\vec{\mathbf w}) \cdot \frac{1}{(N-n)!} \int_{\mathcal T^{N-n}} \bigg\{ \prod_{\ell<m}^{N-n} \Delta_{\mathcal U}(\mathbf x_m, \mathbf x_{\ell})^{\beta}\bigg\} \, d\mu_{\mathcal U}^{N-n}(\vec{\mathbf x}).
\]
Notice also, that
\[
Z_{N,\mathcal U}(\beta) = \sum_{n=0}^{N} \frac{1}{n!} \int_{\mathcal P^n} \bigg\{ \prod_{j<k}^n \Delta_{\mathcal U}(\mathbf w_k, \mathbf w_j)^{\beta}\bigg\} \, d\mu_{\mathcal U}^n(\vec{\mathbf w}) \cdot \frac{1}{(N-n)!} \int_{\mathcal T^{N-n}} \bigg\{ \prod_{\ell<m}^{N-n} \Delta_{\mathcal U}(\mathbf x_m, \mathbf x_{\ell})^{\beta}\bigg\} \, d\mu_{\mathcal U}^{N-n}(\vec{\mathbf x}).
\]
We conclude that,
\begin{align*}
0 &= \sum_{n=0}^N (\E[N_{\mathcal P}] - n) \frac{(q+1)^{-n}}{n!} \int_{\mathcal P^n} \bigg\{ \prod_{j<k}^n \delta_{\mathcal U}(\mathbf w_k, \mathbf w_j)^{\beta}\bigg\} \, d\mu_{\mathcal U}^n(\vec{\mathbf w}) \cdot \frac{(q+1)^{n-N}}{(N-n)!} \int_{\mathcal T^{N-n}} \bigg\{ \prod_{\ell<m}^{N-n} \delta_{\mathcal U}(\mathbf x_m, \mathbf x_{\ell})^{\beta}\bigg\} \, d\mu_{\mathcal U}^{N-n}(\vec{\mathbf x}) \\
&= \sum_{n=0}^N (\E[N_{\mathcal P}] - n) \frac{q^{-\beta{n \choose 2}-n}}{n!} \int_{\mathcal P^n} \bigg\{ \prod_{j<k}^n \delta_{\mathcal P}(\mathbf w_k, \mathbf w_j)^{\beta}\bigg\} \, d\mu_{\mathcal P}^n(\vec{\mathbf w}) \cdot \int_{\mathcal T^{N-n}} \bigg\{ \prod_{\ell<m}^{N-n} \delta_{\mathcal T}(\mathbf x_m, \mathbf x_{\ell})^{\beta}\bigg\} \, d\mu_{\mathcal T}^{N-n}(\vec{\mathbf x}),
\end{align*}
and the theorem follows.
\end{proof}
This theorem is particularly interesting when $\mathcal P$ is related to $\mathcal T$.

\begin{ex}
When $\mathcal T$ and $\mathcal P$ are copies of the profinite completion of $\Z_q$ (alternately described as infinite, rooted $q$-nary trees), then the regularity implies that the coordinates of $\mathbf x \in \mathcal U^N$ are equally distributed among the children of $v_0$, of which there are $q+1$. It follows that in this situation, $\E[N_{\mathcal P}] = N/(q+1)$, and thus
\[
Z_N(\beta) = \frac{1}{N!}\int_{\Z_q^N} \bigg\{ \prod_{m<n} |x_n - x_m|_q^{\beta} \bigg\} \, d\mathbf x,
\]
satisfies the quadratic recurrence
\[
\sum_{n=0}^N \left(\frac{N}{q+1} - n\right) q^{-\beta{n \choose 2}-n} Z_n(\beta) \, Z_{N-n}(\beta) = 0,
\]
as was first reported in \cite{sinclair2020nonarchimedean}. In the $p$-adic situation, glueing in another child tree to the root, is equivalent to gluing in another coset of $p \Z_p$ into $\Z_p$. This can be seen to be topologically equivalent to projective space over $\Q_p$, and the electrostatics in that situation were investigated in \cite{MR4586932} and similar recurrences appear there.
\end{ex}

\subsection{Quadratic Recurrences for Mean Partition Functions}

The quadratic recurrences that arise in the previous section arise because when we glue $\mathcal P$ into $\mathcal T$ as a subtree to create $\mathcal U$, the conditional Boltzmann factor, conditioned on the number of coordinates in $\mathcal P$ factors into the Boltzmann factors over $\mathcal T$ and $\mathcal U$. This happens because the potential energy of $\mathbf x \in \mathcal T$ and $\mathbf w \in \mathcal P$ is given by $-\log \Delta_{\mathcal U}(\mathbf x, \mathbf w) = 0$ by construction. That is, conditioned on $N_{\mathcal P} = n$ the distribution of
$\vec{\mathbf x} \in \mathcal T^{N-n}$ and $\vec{\mathbf w} \in \mathcal P^n$ are independent. Notice however, that the energy {\em does} depend on $N_{\mathcal P}$, because putting more particles into $\mathcal P$ `costs' more energy due to the $1/Q$ factor that appears in front of $\delta_{\mathcal P}$ in the formula for $\Delta_{\mathcal U}$.

Let's try to replicate this when $\mathcal T$ and $\mathcal P$ are random (perhaps eventually with the same law, but for now our assumptions are minimal). Let us, in this case, construct $\mathcal U$ by connecting the root $v_0$ of $\mathcal T$ and the root $w_0$ of $\mathcal P$ to a new vertex $u_0$ which we will view as the root of $\mathcal U$. We will define the ultrametric
\[
\Delta_{\mathcal U}(\mathbf x, \mathbf y) = \piecewise{1 & \mathbf x \in \mathcal T, \mathbf y \in \mathcal P; \\
\delta_{\mathcal T}(\mathbf x,\mathbf y) & \mathbf x, \mathbf y \in \mathcal T; \\
\delta_{\mathcal P}(\mathbf x,\mathbf y) & \mathbf x, \mathbf y \in \mathcal P.}
\]
Our assumptions will be that we have chosen $\mathcal T$ and $\mathcal P$ randomly and independently, and then we have constructed an ultrametric on $\mathcal U$ whose negative logarithm gives the energy contribution of pairs of particles. We will also introduce a quantity $E_n$ which represents the energy cost of having $n$ particles in $\mathcal P$. There are natural choices for $E_n$, for instance, $E_n = c n$ represents the situation where each particle costs energy $c$ which will be explored further in the section, for now we maintain some generality. Thus, the energy functional on $\mathcal U$ conditioned on $N_{\mathcal P} = n$ is given by
\[
E(\vec{\mathbf u}) = -E_n - \sum_{j<k}^n \log \Delta_{\mathcal U}(\mathbf u_j, \mathbf u_k).
\]
Boltzmann statistics imply then that the mean partition function for the random ensemble is given by
\[
\overline Z_{N,\mathcal U}(\beta) = \E\left[\sum_{n=0}^N e^{-\beta E_n} \frac1{(N-n)!}\int_{\mathcal T^{N-n}} \bigg\{ \prod_{\ell<m} \delta_{\mathcal T}(\mathbf x_m, \mathbf x_{\ell})^{\beta}\bigg\} d\mu^{N-n}_{\mathcal U}(\vec{\mathbf x}) \cdot \frac1{n!} \int_{\mathcal P^n} \bigg\{ \prod_{j<k} \delta_{\mathcal P}(\mathbf u_k, \mathbf u_j)^{\beta}\bigg\} d\mu_{\mathcal U}(\vec{\mathbf u})\right],
\]
where the expectation is over all $\mathcal T$ and $\mathcal P$. Note that, because $u_0$ has exactly two children, on $\mathcal P$, $2 \mu_{\mathcal U} = \mu_{\mathcal P}$ and on $\mathcal T$, $2 \mu_{\mathcal U} = \mu_{\mathcal T}$. Note also that, if we assume $\mathcal P$ is independent of $\mathcal T$, then so too are their partition functions (this is still true after conditioning on $N_{\mathcal P}$)
\[
\overline Z_{N,\mathcal U}(\beta) = 2^N \sum_{n=0}^N e^{-\beta E_n} \overline Z_{n, \mathcal P}(\beta) \overline Z_{N-n, \mathcal T}(\beta) .
\]
Similarly, the expectation of $N_{\mathcal P}$ is given by
\[
\E[N_{\mathcal P}] = \frac{2^N}{\overline Z_{N,\mathcal U}(\beta)} \sum_{n=0}^N n e^{-\beta E_n} \overline Z_{n, \mathcal P}(\beta) \overline Z_{N-n, \mathcal T}(\beta) .
\]
Cross multiplying and expanding, we have the following theorem.
\begin{thm}
\label{thm:random quadratic recurrence}
\[
0 = \sum_{n=0}^N \left(\E[N_{\mathcal P}] - n \right) e^{-\beta E_n} \overline Z_{n, \mathcal P}(\beta) \overline Z_{N-n, \mathcal T}(\beta) .
\]
\end{thm}
Theorem~\ref{thm:random quadratic recurrence} is most interesting when $\mathcal T$ and $\mathcal P$ are related (for instance independent copies of the same tree) and $E_n$ is chosen such that $\E[N_{\mathcal P}]$ can be computed.

\begin{ex}
It is sometimes easy to engineer a situation where $\E[N_{\mathcal P}]$ can be computed explicitly, due to some symmetry in the system, but for which Theorem~\ref{thm:random quadratic recurrence} is vacuous, and hence does not produce a recurrence. For instance, if $\mathcal P$ and $\mathcal T$ are independent and identically distributed, and if we assume all $E_n = 0$ (the latter implies that there is no preference for a particle to be in $\mathcal P$ vs $\mathcal T$), then $\E[N_{\mathcal P}] = N/2$, and the coefficient in front of $Z_{N, \mathcal T}$ is 0.
\end{ex}

\begin{ex}
Suppose $\mathcal T$ and $\mathcal P$ are independent branching processes with the same law determined by the non-negative integer-valued random variable $Q$. Suppose further that $\E[Q]$ is finite. Can we conclude that when $E_n = {n \choose 2} \log \E[Q]$ that $\E[N_{\mathcal P}] = N/(\E[Q]+1)$? Does this then imply that
\[
0 = \sum_{n=0}^N \left(\frac{N}{\E[Q]+1} - n \right) \E[Q]^{-\beta {n \choose 2}} \overline Z_{n, \mathcal T}(\beta) \overline Z_{N-n, \mathcal T}(\beta)?
\]
It is tempting to try to prove this by taking expectation of Theorem~\ref{thm: quadratic recursion}, however when we do that, we see
\[
  \sum_{n=0}^N \E\left[(\E[N_{\mathcal P}] - n) Q^{-\beta{n \choose 2}-n} Z_{\mathcal T, N-n}(\beta) Z_{\mathcal P,n}(\beta) \right] = 0,
\]
and we can't proceed (in general) because the inner expectation $\E[N_{\mathcal P}]$ is an expectation over the states of $\mathcal U$ and hence depends on $\mathcal T$, $\mathcal P$ and how they are stitched together to form $\mathcal U$. Even if this obstacle were surmountable, there is the additional complication that, under the assumptions of Theorem~\ref{thm: quadratic recursion}, $\mathcal T$ and $Q$ are decidedly {\em not} independent.
\end{ex}

\section{The Grand Canonical Ensemble}

In this situation we allow $N$ to take any non-negative value by putting $\mathcal T$ in contact with a structureless energy and particle reservoir so that energy and particles can flow to and from $\mathcal T$ and the energy functional conditioned on the number of particles $N$ given by
\[
E(\vec{\mathbf x}) = -E_N -\sum_{m < n} \log \delta_{\mathcal T}(\mathbf x_n, \mathbf x_m); \qquad \vec{\mathbf x} \in \mathcal T^N.
\]
The {\em grand canonical ensemble} is that where each particle costs the same amount of energy $c$ (called the {\em chemical potential}). That is the grand canonical ensemble is when $E_N = c N$, and the conditional probability, given $N$ particles, that the system is in state $\vec{\mathbf x} \in \mathcal T^N$ is specified by
\[
\frac{e^{\beta c N}}{N! Z(c, \beta)} \bigg\{\prod_{m<n} \delta_{\mathcal T}(\mathbf x_n, \mathbf x_m)^{\beta}\bigg\} d\mu_{\mathcal T}^N(\vec{\mathbf x}),
\]
where
\[
Z(c, \beta) = \sum_{N=0}^{\infty} e^{\beta c N} \frac{1}{N!} \int_{\mathcal T^N} \bigg\{\prod_{m<n} \delta_{\mathcal T}(\mathbf x_n, \mathbf x_m)^{\beta}\bigg\} d\mu_{\mathcal T}^N(\vec{\mathbf x}) = \sum_{N=0}^{\infty} e^{\beta c N} Z_N(\beta)
\]
is the {\em grand canonical} partition function. At times it is useful to make the substitution $t = e^{\beta c}$, in which case we write
\[
Z(t, \beta) = \sum_{N=0}^{\infty} t^N Z_N(\beta),
\]
which is also the generating function for the $Z_N(\beta)$. The quantity $t$ is sometimes called the {\em fugacity}. In the most general case, we write $\mathbf E = (E_N)$ for the sequence of energy costs for the system having $N$ particles, and we write
\[
Z(\mathbf E, \beta) = \sum_{N=0}^{\infty} e^{\beta E_N} Z_N(\beta).
\]
It should be clear from context which partition function we are working with.

\subsection{The Mean Grand Canonical Partition Function}
The grand canonical partition function depends on $\mathcal T$, and if we want to make $\mathcal T$ random, we need to investigate the {\em mean} grand canonical partition function. When $\mathcal T$ is the profinite completion of the branching process things simplify dramatically.

To codify this, let us define the generating functions for $\{Z_N(\beta) : N \in \N\}$,
\[
F_q(\beta, t) = \sum_{N=0}^{\infty}  Z_N(\beta) q^{-\beta{N \choose 2}} t^N; \quad q=1,2,\ldots.
\]
Note that $Z(\beta, t) = F_1(\beta, t)$ is the ``ordinary'' generating function for $\{Z_N(\beta) : N \in \N\}$. Note also, that because $Z_N(\beta)$ is random (it depends on $\mathcal T$), so too are the $F_q(\beta, t)$. The associated deterministic {\em mean} generating functions are then defined to be
\[
\overline F_q(\beta, t) = \sum_{N=0}^{\infty}  \overline Z_N(\beta) q^{-\beta{N \choose 2}} t^N; \quad q=1,2,\ldots.
\]
In particular $\overline Z(\beta, t) = \overline F_1(\beta, t)$ is the {\em mean} grand canonical partition function.
\begin{thm}
\label{thm:1}
Suppose $\mathcal T$ is the profinite completion of a random branching process with law determined by $Q$, then
\begin{align*}
\overline Z(\beta, t) &= \mathbb E\left[\left(\overline F_Q\left(\beta, \frac{t}Q\right)\right)^Q \right] \\
&= \mathbb E\left[\left(\overline F_Q\left(\beta, \frac{t}Q\right)\right)^Q \bigg| Q > 1\right].
\end{align*}
\end{thm}

\begin{proof}
\begin{align*}
\overline Z(\beta, t) := \sum_{N=0}^{\infty} \overline Z_N(\beta) t^N &= \sum_{N=0}^{\infty} t^N \sum_{q} p_q \sum_{N_1 + \cdots + N_q = N} \prod_{k=1}^q q^{-N_k -\beta{N_k \choose 2}} \overline Z_{N_k}(\beta) \\
&= \sum_q p_q \sum_{N_1, \ldots, N_q} \prod_{k=1}^q t^{N_k} q^{-N_k -\beta{N_k \choose 2}} \overline Z_{N_k}(\beta) \\
&= \sum_q p_q \left( \sum_{N=0}^{\infty} \overline Z_N(\beta) q^{-\beta{N \choose 2}} \left(\frac{t}q\right)^N  \right)^q. \qedhere
\end{align*}
\end{proof}

\begin{ex}
When $\beta = 0$, we have $\overline Z_N(0) = Z_N(0) = \frac{1}{N!}$, and hence $\overline F_q(0, t) = F_q(0, t) = e^t$. We immediately verify that
\[
e^t = \overline Z(0,t) = \sum_q p_q (e^{t/q})^q = e^t,
\]
as claimed.
\end{ex}

\begin{ex}
As $\beta \rightarrow +\infty$,
\[
\overline Z(+\infty, t) = \sum_q p_q \left(\overline F_q(+\infty, t/q)\right)^q,
\]
and since $\overline Z(+\infty, t) = \overline F_1(+\infty, t)$ appears on the right-hand-side, we may solve for it and find
\[
\overline Z(+\infty, t) = \frac1{1 - p_1}\sum_{q>1} p_q \left(\overline F_q(+\infty, t/q)\right)^q.
\]
If $q > 1$ then $q^{-\beta{N \choose 2}} \rightarrow 0$ unless $N=0$ or $1$ (for which we interpret ${0 \choose 2} = {1 \choose 2} = 0$) and hence for $q > 1$,
\[
\overline F_q(+\infty, t/q) = 1 + \frac{t}{q}
\]
It follows that
\[
\overline Z(+\infty, t) = \frac1{1 - p_1}\sum_{q>1} p_q \left(1 + \frac{t}{q}\right)^q = \E\bigg[\left(1 + \frac{t}Q \right)^Q \bigg| Q > 1 \bigg].
\]
When $Q = q$ with probability $p$ and $Q = 1$ with probability $1-p$, then
\[
\overline Z(+\infty, t) = \left(1 + \frac{t}{q} \right)^q,
\]
which agrees with the electrostatics in $\mathbb Z_q$ situation.
\end{ex}

One way to contextualize Theorem~\ref{thm:1} is to define $\Xi : \C \times \C[[t]] \rightarrow \C[[t]]$ specified by
\[
\Xi\bigg(z, \sum_N a_N t^N \bigg) = \sum_N a_N z^{N \choose 2} t^N.
\]
We may view $\Xi(z, \cdot)$ as an operator on the ring of formal power series $\C[[t]]$, and to connect this with previous notation,
\[
\overline F_q(\beta, t) = \Xi\big(q^{-\beta}, \overline Z(\beta, t) \big).
\]
This allows us to reformulate Theorem~\ref{thm:1}.
\begin{cor}
$\overline Z(\beta, t)$ is the unique fixed point fixed point in $\C[[t]]$ of the operator
\[
\alpha(t) \mapsto \sum_q p_q \Xi\big(q^{-\beta},\alpha(t/q)\big)^q = \E\big[\Xi\big(Q^{-\beta}, \alpha(t/Q)\big)^Q\big] = \E\big[\Xi\big(Q^{-\beta}, \alpha(t/Q)\big)^Q \big| Q > 1\big].
\]
satisfying $\alpha(0) = \alpha'(0) = 1$.
\end{cor}

\bibliography{bibliography}

\begin{center}
\noindent\rule{4cm}{.5pt}
\vspace{.25cm}

\noindent {\sc \small Christopher D.~Sinclair}\\
{\small Department of Mathematics, University of Oregon, Eugene OR 97403} \\
email: {\tt csinclai@uoregon.edu}
\end{center}

\end{document}